\documentclass[submission,copyright,creativecommons]{eptcs}
\providecommand{\event}{Physics and Computation 2016}
\usepackage{breakurl}   
\usepackage{amssymb,amsmath,amsthm,color,cite}
\usepackage{hyperref}

\newtheorem{theorem}{Theorem}[section]
\newtheorem{proposition}[theorem]{Proposition}
\theoremstyle{definition}
\newtheorem{definition}[theorem]{Definition}
\newtheorem{example}[theorem]{Example}

\title{Physical Computation, $P/poly$ and $P/log\star$}
\author{Richard Whyman
\institute{The University of Leeds,\\
Leeds, UK}
\email{mmrajw@leeds.ac.uk}
}
\def\titlerunning{Physical Computation, $P/poly$ and $P/log\star$}
\def\authorrunning{Richard Whyman}
\begin{document}
\mathcode`@="8000 
{\catcode`\@=\active\gdef@{\mkern1mu}}
\maketitle
\setcounter{tocdepth}{2}
\newcounter{ThmNum}
\newcounter{ThmNum2}
\newcounter{ThmNum3}
\begin{abstract} 
In this paper we give a framework for describing how abstract systems can be used to compute if no randomness or error is involved. Using this we describe a class of classical ``physical" computation systems whose computational capabilities in polynomial time are equivalent to $P/poly$. We then extend our framework to describe how measurement and transformation times may vary depending on their input. Finally we describe two classes of classical ``physical" computation systems in this new framework whose computational capabilities in polynomial time are equivalent to $P/poly$ and $P/log\star$.
\end{abstract}

\section{Introduction}
To answer the question of exactly how much various physical systems are capable of computing, we must first have a good abstract description of them, balancing simplicity with descriptive power. This paper discusses such a description.

In \cite{WhenDoesAPhysicalSystemCompute} Horsman \emph{et al.} argue convincingly that, in general, we cannot make use of a physical system for computation unless we have a way of interacting with it that allows us to predict the nature of its output given a particular input. Building on this idea, we will define the systems we wish to compute from in terms of how we can interact with them.

In \cite{BeggsNewtonian} Beggs \emph{et al.} describe how a Newtonian kinematic system can be used to tackle a problem that's uncomputable for a Turing machine; computing the characteristic function of any given subset of $\mathbb{N}$. Similarly, they achieve oracle-like results using  ``experiments" consisting of either a precise set of scales \cite{beggs2013oracles}, or of a cannon and a wedge \cite{beggs2012impact}, calling a Turing machine combined with such classical physical experiments an ``analogue-digital device".

A key aspect discussed in Beggs \emph{et al.}'s later papers is that the time taken for their examples of oracle-like queries to be carried out must depend on what is being queried, thereby restricting the speed with which certain problems can be solved. This culminated in Beggs \emph{et al.}'s ``analogue-digital Church-Turing thesis" \cite{beggs2014analogue}, which states: ``no possible abstract analogue-digital device can have more computational capabilities in polynomial time than $BBP//log\star$." Though this thesis was very well justified, Beggs \emph{et al.} (perhaps wisely) avoided giving a formal mathematical description of these analogue-digital devices. Here we shall attempt to give such a description. However, whilst in the thesis the analogue-digital devices always have finite, possibly unbounded, precision to their actions, due to complications that arise from exactly how this error is treated,\footnote{For example, how do we describe the outcome of applying a transformation $T$ with an error bound of $\epsilon$? If we take it to be a probability distribution then what distributions and probability measures are appropriate?} 
we shall avoid defining inexact systems and instead focus here on giving a robust framework for describing computation on systems without any error or randomness, with the hope that this will eventually lead to a more general and inclusive framework. 

In \cite{beggs2014analogue} Beggs \emph{et al.} suggested that the class of problems solvable by analogue-digital devices with infinite precision (and therefore without error) in polynomial time is likely $P/log\star$ and\footnote{This notation is explained in Definitions~\ref{NUCC} and~\ref{NUCC2}.} at most $P/poly$. In this paper we prove that the $P/poly$ result is true for a restricted class of error-free systems that we believe are physically reasonable. We then apply an additional restriction to these systems to obtain the $P/log\star$ result. 

\section{Computation Systems}
The goal of almost any task is to obtain something from something else. Hence we shall build up our framework for describing physical computation by beginning with an abstract description of a system that can manipulate some space via a fixed set of operations and observe it via a fixed set of measurements\footnote{Note that these measurements will differ from quantum measurements in the sense that what is being observed is not necessarily altered by the measurement.}.

\begin{definition}
	A \textsl{computation system} is a quadruple $C = (X,\Pi,\mathcal{T},x_0)$ where:
	\begin{itemize}
	\item $X$ is a non-empty set,
	\item $\Pi$ is a finite non-empty set of finite partitions\footnote{A partition $\alpha$ of a set $X$, is a set of disjoint subsets of $X$ such that $\bigcup_{A \in \alpha} A = X$. It is finite if $|\alpha|$ is finite.} of $X$,
	\item $\mathcal{T}$ is a finite non-empty set of transformations $T: X \rightarrow X$ and
	\item $x_0 \in X$.
	\end{itemize}
\end{definition}
The idea behind this quadruple is as follows. The set $X$ describes the possible configurations that a device acting upon the system can be in. 
The set $\Pi$ describes the set of measurements that can be performed in a single step by such a device without altering its configuration. 
The set $\mathcal{T}$ describes the set of actions that the device can perform in a single step to alter its configuration. The point $x_0$ is the configuration that the device begins in.

A finite partition of $X$ can be regarded as describing measurement of it, as any measurement of $X$ is essentially a process assigning a particular value to each element of $X$. In other words, as in rough set theory \cite{komorowski1999rough} we have some attribute valuation $V: X \rightarrow \{1,\hdots,n\}$. If we then define an equivalence relation $\sim$ on $X$ by letting $x_1 \sim x_2$ iff $V(x_1) = V(x_2)$, then the equivalence class generated by this is a finite partition on $X$. So if $\alpha = X/\!\!\sim$ and $A \in \alpha$ then $A = \{x \in X \text{ } | \text{ } V(x) = a\}$ for some $a \in \{1,\hdots,n\}$.

\begin{example}
	Let $\mathcal{A}$ be an alphabet, and $\mathbf{B} \not\in \mathcal{A}$ be the blank tape symbol. So let $\mathcal{A}_\mathbf{B} = \mathcal{A} \cup \{\mathbf{B}\}$, then $\mathcal{A}_\mathbf{B}^\mathbb{Z}$ is the set of bi-infinite strings from $\mathcal{A}_\mathbf{B}$ and the set of possible Turing tape configurations for a Turing machine with tape alphabet $\mathcal{A}$ is: $$\mathcal{A}^{TT} = \bigcup_{m,n \in \mathbb{Z}}\{(x_i)_{i \in \mathbb{Z}} \in \mathcal{A}_\mathbf{B}^\mathbb{Z} \text{ } | \text{ } x_i = \mathbf{B} \text{ if } i < m \text{ or } i > n\}.$$
	For any $a \in \mathcal{A}_\mathbf{B}$, let $\langle a \rangle = \{(x_i)_{i \in \mathbb{Z}} \in \mathcal{A}^{TT} \text{ } | \text{ } x_0 = a\}$ then the tape reading partition is: $$R(\mathcal{A}) = \{\langle a \rangle \text{ } | \text{ } a \in \mathcal{A}_\mathbf{B}\}.$$ 
	Define $\lhd$ and $\rhd$ to be the left and right shifts of $\mathcal{A}^{TT}$ respectively and for each $a \in \mathcal{A}_\mathbf{B}$ let $r_a$ be the ``replace with $a$" operation, that is if $(w_i)_{i \in \mathbb{Z}} = r_a((x_i)_{i \in \mathbb{Z}})$ then $w_0 = a$ and $w_i = x_i$ if $i \in \mathbb{Z} \setminus \{0\}$. The set of Turing tape transformations on $\mathcal{A}^{TT}$ is then: $$W(\mathcal{A}) = \{I_{\mathcal{A}^{TT}},\lhd,\rhd\} \cup \{r_a \text{ } | \text{ } a \in \mathcal{A}_\mathbf{B}\},$$ where $I_\mathcal{A}^{TT}$ denotes the identity map on $\mathcal{A}^{TT}$. Finally, let $\mathbf{B}^\mathbb{Z} = (b_i)_{i \in \mathbb{Z}}$ where $b_i = \mathbf{B}$ for all $i \in \mathbb{Z}$.

	Then a Turing machine with an alphabet of $\mathcal{A}$ can be described as the computation system $TM(\mathcal{A}) = (\mathcal{A}_\mathbf{B}^\mathbb{Z},\{R(\mathcal{A})\},W(\mathcal{A}),\mathbf{B}^\mathbb{Z})$.
\end{example}

We can also describe an oracle Turing machine with oracle $A \subseteq \mathcal{A}^*$ as a computation system by taking a Turing machine with two tapes and giving the second an additional transformation $O_{A,a}: \mathcal{A}_\mathbf{B}^\mathbb{Z} \rightarrow \mathcal{A}_\mathbf{B}^\mathbb{Z}$. Where for some element $a \in \mathcal{A}$, the transformation $O_{A,a}$ changes the symbol the second tape head is pointing at to $a$ if the finite string of non-blank tape symbols to the right of the head is a word in $A$, otherwise it changes it to $\mathbf{B}$.

We shall allow a device acting on $C$ to also have access to a Turing machine, rather than acting solely on $C = (X,\Pi,\mathcal{T},x_0)$. This is because the computational power of a computation system be would severely, and arguably, unnecessarily restricted by an inability to perform operations that are trivial to a Turing machine. For example recording or copying the information obtained during a computation. Further, we want the inputs into our system to be finitely knowable objects, such as those described by finite words, so the Turing tape is where the inputs can be put in to our device.

As we shall see below, the manner in which we describe how a device acts upon a computation system is similar to how a Turing machine acts upon a Turing tape. Its actions are defined by a set of rules whose implementation is also dependant on an set of internal states.

\begin{definition}
	A \textsl{set of states} is a finite set $S$ containing at least three elements $s_0$, $s_a$ and $s_r$, called the \textsl{initial}, \textsl{accepting} and \textsl{rejecting} states respectively.
\end{definition}

The initial state $s_0$ is the internal state that a device always starts in. For any computation the device will always halt if it reaches either $s_a$ or $s_r$.

\begin{definition}
	Let $\mathcal{A}$ be an alphabet and $S$ be a set of states. An \textsl{$\mathcal{A}^*$-program} on $C = (X,\Pi,\mathcal{T},x_0)$ and $S$ is a finite set of rules $Q$ which describes how a device acting on $C$ and a Turing machine with tape alphabet $\mathcal{A}$ behaves. Each rule takes the form $(s_i,\alpha,A,s_j,T)$ where $s_i \in S \setminus \{s_a,s_r\}$, $s_j \in S$, $\alpha \in \Pi \cup \{R(\mathcal{A})\}$, $A \in \alpha$ and $T \in \mathcal{T} \cup W(\mathcal{A})$, this rule can be read as ``if the device is in state $s_i$, perform measurement $\alpha$, then if it is at a point in the subset $A$ go to state $s_j$ and perform action $T$". 
\end{definition}

Mathematically this means the following. Suppose that the device is at the configuration $x \in X$ with $w \in \mathcal{A}^{TT}$ written on its tape and an internal state of $s_i$, then if there is some rule in $Q$ beginning with $(s_i,\alpha)$ the device then performs an $\alpha$ measurement. There are two cases; if $\alpha \not= R(\mathcal{A})$ then if it is the case that $x \in A$ and there is a rule $(s_i,\alpha,A,s_j,T) \in Q$ then this rule is applied.
If instead $\alpha = R(\mathcal{A})$, then if it is the case that $w \in B$ and there is a rule $(s_i,R(\mathcal{A}),B,s_k,U) \in Q$ then this rule is applied. In either of the above cases, if no appropriate rule exists then the internal state becomes $s_r$. If the internal state becomes either $s_a$ or $s_r$ then the device halts.

Applying the rule $(s_i,\alpha,A,s_j,T)$ to a device at $(x,w,s_i)$ results in it becoming $(T(x),w,s_j)$ if $T \in \mathcal{T}$ and $(x,T(w),s_j)$ if $T \in W(\mathcal{A})$.

In order for the above process to be deterministic we require that all rules beginning with the same internal state must have the same partition, so for any $(s_i,\alpha,A,s_j,T),(s_k,\beta,B,s_l,U) \in Q$, if $s_i = s_k$ then $\alpha = \beta$. We also require that if two rules begin with the same state, partition and subset, then they must also end with the same state and transformation so for any $(s_i,\alpha,A,s_j,T),(s_k,\beta,B,s_l,U) \in Q$, if $(s_i,\alpha,A) = (s_k,\beta,B)$ then $(s_j,T) = (s_l,U)$.

\begin{definition}
	A device implementing an $\mathcal{A}^*$-program $Q$ on a computation system $C$ takes as its input a word $w \in \mathcal{A}^*$, written onto its Turing tape as the configuration $w^\dagger$. It then repeatedly applies the rules of $Q$ to $(x_0,w^\dagger,s_0)$ until it either reaches the internal state $s_a$ and ``accepts" $w$ or it reaches the internal state $s_r$ and ``rejects" $w$. If $w$ is accepted we write $\varphi^C_{Q}(w) = \mathbb{T}$ and if it is rejected we write $\varphi^C_{Q}(w) = \mathbb{F}$. Otherwise, if neither $s_a$ nor $s_r$ is ever reached, then the computation never ends and $\varphi^C_{Q}(w)$ is undefined.
\end{definition}

\begin{definition}
	Let $C = (X,\Pi,\mathcal{T},x_0)$ be a computation system. A subset\footnote{Which we shall sometimes refer to as a problem.} $A \subseteq \mathcal{A}^*$ is \textsl{computable using $C$} if there exists an $\mathcal{A}^*$-program $Q$ on $C$ such that for any $w \in \mathcal{A}^*$:
	\begin{align*}
	\varphi^C_{Q}(w) = \mathbb{T} &\iff w \in A, \\
	\varphi^C_{Q}(w) = \mathbb{F} &\iff w \not\in A.
	\end{align*}
\end{definition}

It is worth noting that the partitions of $\Pi$ do not fully describe the information a device acting on $C = (X,\Pi,\mathcal{T},x_0)$ can extract from $X$, as the device can also transform $X$ by some $T \in \mathcal{T}$. Indeed if we know that $x \in A$ and $T(x) \in B$ then we also know that $x \in A \cap T^{-1}(B)$.

\begin{example}\label{NES}
	Consider the computation system of the form $C_\phi = ([0,1),\{\alpha\},\{T\},\phi)$, where we have $\alpha = \{[0,\frac{1}{2}),[\frac{1}{2},1)\}$ and $T(x) = 2x - \lfloor 2x \rfloor$ for any $x \in [0,1)$. We can then compute the binary expansion of the number $\phi$ to arbitrarily many places using $C_\phi$.
	To do this, take an $\{0,1\}^*$-program $Q$, which firstly takes an $\alpha$ measurement of the starting number $\phi$. Either $\phi \in [0,\frac{1}{2})$ and the first binary digit of $\phi$ is $0$, or $\phi \in [\frac{1}{2},1)$ and the first binary digit of $\phi$ is $1$. So $Q$ records this result on the Turing tape before applying the transformation $T$ to $\phi$, which effectively deletes its first binary digit. This process can then repeated, with $Q$ taking an $\alpha$ measurement of $T(\phi)$ to obtain the second digit of $\phi$, and so on.
\end{example}

When defining the set of transformations on a computation system we could have chosen for our transformations to be continuously applicable, meaning that the value of $T^z(x)$ varies continuously with $z \in [0,\infty)$ and $T^{z_2} \circ T^{z_1}(x) = T^{z_1+z_2}(x)$ for $z_1,z_2 \in [0,\infty)$. We consider our devices to be acting in conjunction with Turing machines, as these work in discrete time we can only really consider transformations that are implemented discretely. However objects in our own world can only be altered through continuous processes, to account for this in such a scenario we take $T(x)$ to be the result of continuously applying a process given by $T$ to $x$ for a fixed multiple of a single Turing machine time step. As multiplying the time taken by a constant factor does not change the usual complexity classes we have the following notion for the computation time of computation systems.

\begin{definition}
	A \textsl{time function} for an $\mathcal{A}^*$-program $Q$ on a computation system $C = (X,\Pi,\mathcal{T},x_0)$ is a function $t: \mathbb{N} \rightarrow \mathbb{N}$ such that for any $w \in \mathcal{A}^*$ inputted into the device, the number of times we apply the rules of $Q$ before the computation halts is bounded by $t(|w|)$. Clearly, if $\varphi^C_Q(v)$ is undefined for some $v \in \mathcal{A}^*$ then $Q$ does not have a time function.

	We say that a problem $A \subseteq \mathcal{A}^*$ is \textsl{computable using $C$ in polynomial time} if there exists an $\mathcal{A}^*$-program $Q$ on $C$ with a polynomial time function.
\end{definition}

Computing $n$ digits of a number $\phi$ using the computation system $C_\phi$ in Example~\ref{NES} takes at least $n$ rule applications. There are $2^k$ words of length $k$ in $\{0,1\}^*$, so whilst we could encode the characteristic function of a arbitrary set $A \subseteq \{0,1\}^*$ in $\phi$, computing the membership of $A$ using $C_\phi$ would in general be impossible in polynomial time.

However, restricting computation systems to being only able to solve problems in polynomial time is not, in general, a restriction at all. All oracle Turing machines are examples of computation systems, hence the class of all problems computable by computation systems in polynomial time is just the class of all problems.

\section{Classical Physical Computation Systems and $P/poly$}
\setcounter{ThmNum2}{-1}
Suppose we have some system of objects. Classical physics each of these objects possess a set of quantities (position, momentum, velocity etc.) which can each be described by some real number in some particular dimension. The dimension of any physical quantity is expressed as a product of the basic physical dimensions (length, mass, time etc.) each raised to an integer power. It does not make sense to take the exponent of a physical quantity but we can multiply and divide by arbitrary physical quantities. Adding together two physical quantities is possible if they have the same dimension. We can also take the $n$th root of a physical quantity if its basic physical dimensions are all multiples of $n$.

In order to manipulate the objects of a classical physical system we manipulate their physical quantities. This suggests that the transformations we are able to perform on a classical physical system must be able to send physical quantities to physical quantities, and thus must be built up from some finite composition of additions, multiplications and rational powers.

\begin{definition}\label{PTD}
	A \textsl{multi-variable polynomial function with rational powers} on $\mathbb{R}^m$ is a function $F: \mathbb{R}^m \rightarrow \mathbb{R}$ such that for some fixed $I \in \mathbb{N}$, $r_1,\hdots,r_I \in \mathbb{R}$ and $q_{11},\hdots,q_{Im} \in \mathbb{Q}$: $$F(x_1,\cdots,x_m) = \sum_{i=1}^I r_i \prod_{j=1}^{m}x_j^{q_{ij}},$$ for any $(x_1,\cdots,x_m) \in \mathbb{R}^m$. We call the numbers $r_1,\hdots,r_I$ the coefficients of $F$. In the cases where $q_{ij} = \frac{a}{b}$ and $b \not= 1$ the function $f(x) = x^{q_{ij}}$ may either have two real roots, in which case we take $x^{q_{ij}}$ to be the greatest of these roots, or it may have zero real roots, in which case we take $x^{q_{ij}}$ to be undefined. 

	We call a transformation $T:\subseteq \mathbb{R}^m \rightarrow \mathbb{R}^m$ is a \textsl{classical transformation} on $\mathbb{R}^m$ if there are multi-variable polynomial functions with rational powers $F_1,\hdots,F_m$ such that:
	$$T\begin{pmatrix}x_1 \\ \vdots \\ x_m\end{pmatrix} = \begin{pmatrix}F_1(x_1,\cdots,x_m) \\ \vdots \\ F_m(x_1,\cdots,x_m)\end{pmatrix},$$ 
	for any $x_1,\hdots,x_m \in \mathbb{R}$. We denote the set of classical transformations on $\mathbb{R}^m$ by $ClT^m$.
\end{definition}

Note that whilst transcendental functions such as $\cos$ and $\sin$ are not constructable as classical transformations, applying a rotation by an angle $\theta$ about the origin to any $(x,y) \in \mathbb{R}^2$ results in $(x\cos(\theta) - y\sin(\theta),x\sin(\theta) + x\sin(\theta))$. As the values of $\cos(\theta)$ and $\sin(\theta)$ are fixed such a map is a classical transformation. Indeed, the application of any $m \times m$ matrix is a classical transformation in $ClT^m$.

In order to measure the objects of a classical physical system we are similarly restricted to measuring their physical quantities, which can only be done to some finite degree of accuracy. Knowing that a particle is at a position $\mathbf{x}$ to within an error of $\epsilon$ means knowing it is within the open ball of radius $\epsilon$ centred at $\mathbf{x}$. Knowing that a particle is within a set $U$ as well as a set $V$ means knowing that it is in the set $U \cap V$. Applying a classical transformation $T$ to $\mathbf{x}$ before measuring that it is within the set $A$ and then applying its inverse means knowing this quantity is within the set $T^{-1}(A)$.

\begin{definition}
	For any $\mathbf{x} \in \mathbb{R}^m$ and any $\epsilon > 0$, in the Euclidean metric we denote the open ball and closed balls of radius $\epsilon$ centred at $\mathbf{x}$ as $B_\epsilon(\mathbf{x})$ and $\overline{B}_\epsilon$ respectively.
	Let $ClM^1_0 = \{B_\epsilon(\mathbf{x}),\overline{B}_\epsilon(\mathbf{x}) \text{ } | \text{ } \mathbf{x} \in \mathbb{R}, \epsilon \in (0,\infty)\}$ then we define $ClM^m$ inductively on $\mathbb{N}$ as follows:
	\begin{itemize}
	\item $ClM^m_0 = \{B_\epsilon(\mathbf{x}),\overline{B}_\epsilon(\mathbf{x}),U \times V \text{ } | \text{ } \mathbf{x} \in \mathbb{R}^m, \epsilon \in (0,\infty) \text{ and } U \in ClM^{l}, V \in ClM^{m-l} \text{ for some } l  < m\}$,
	\item $ClM^m_{k+1} = \{U \cup V, U \cap V, \mathbb{R}^m \setminus U, T^{-1}(U) \text{ } | \text{ } U,V \in ClM^m_{k} \text{ and } T \in ClT^m \text{ is invertible}\}$,
	\item $ClM^m = \bigcup_{k \in \mathbb{N}} ClM^m_{k}$.
	\end{itemize}
	A subset $X \subseteq \mathbb{R}^m$ is then called \textsl{classically measurable} if $X \in ClM^m$. A finite partition $\alpha$ of $X$ is a  \textsl{classically measurable partition} if every $A \in \alpha$ is in $ClM^m$.
\end{definition}

Rather than always having to apply the same classical transformation to every of the element of $X$, since we are able to determine which element of a classically measurable partition $\alpha$ a point $x \in X$ is in, we are surely capable of applying a different classical transformation for each partition element.

\begin{definition}
	A transformation $T: X \rightarrow X$ is \textsl{classically constructable} if there is some classically measurable partition $\alpha$ of $X$, such that for any $A \in \alpha$, $T|_A \in ClT^m$ and dom$(T) = X$.
\end{definition}

\begin{definition}
	A \textsl{classical physical computation system} (CPCS) is a computation system of the form $C = (X,\Pi,\mathcal{T},x_0)$ where:
	\begin{itemize}
	\item $\Pi$ contains only classically measurable partitions and
	\item $\mathcal{T}$ contains only classically constructable transformations.
	\end{itemize}
\end{definition}

Note that the above conditions imply that $X$ must itself be a classically measurable subset of $\mathbb{R}^m$ for some $m \in \mathbb{N}$, since it is equal to the finite union of classically measurable sets. Denote the class of problems computable using a CPCS in polynomial time by $P_{CPCS}$.

The following example shows that both of the above conditions on CPCS's are necessary for $P_{CPCS}$ to not be just the class of all problems.

\begin{example}
	We can compute any problem $A \subseteq \{0,1\}^*$ in polynomial time using the computation system $C_1 = (\mathbb{R},\{\alpha_A\},\{p,t\},0)$ where: $$\alpha_A = \{x \in \mathbb{R} \text{ } | \text{ } \text{the binary representation of } \lfloor x \rfloor \text{ is } 1w \text{ for some } w \in A\},$$ which is not necessarily classically measurable. The transformations in $C_1$ are $p(x) = x + 1$ and $t(x) = 2x$ for any $x \in \mathbb{R}$, so they are clearly classically constructable. Now if the integer part of a number $y \in \mathbb{R}$ has a binary representation $b_1b_2\cdots b_k$, then the binary representation of $\lfloor t(y) \rfloor$ is $b_1b_2\cdots b_k0$ and the binary representation of $\lfloor p(t(x)) \rfloor$ is $b_1b_2\cdots b_k1$. We can therefore input the word $w \in \{0,1\}^*$ into $\mathbb{R}$ as a number with a binary representation of $1w$ via a linear number of applications of $p$ and $t$ to 0. By applying the measurement $\alpha_A$ onto the resultant number we can determine whether $w \in A$.

	Alternatively consider the computation system $C_2 = (\mathbb{R},\{\{(-\infty,0),[0,\infty)\}\},\{p,t,T\},0)$ where $p$ and $t$ are the same as before and: $$T(x) = \left\{\begin{array}{ll} 1 & \text{if } x \in \alpha_A , \\ -1 & \text{otherwise}. \end{array}\right.$$ This is not necessarily classically constructable since $\alpha_A$ above is not necessarily a classically measurable partition. As above, we can determine whether $w \in A$ in linear time by inputting it into $\mathbb{R}$ as a number with a binary representation of $1w$ and then applying the transformation $T$ and the measurement $\{(-\infty,0),[0,\infty)\}$ to it. This works as $T^{-1}(\{(-\infty,0),[0,\infty)\}) = \alpha_A$.
\end{example}

We can relate $P_{CPCS}$ to the non-uniform complexity class $P/poly$, which is often defined in terms of boolean circuits, but can also be defined as follows.

\begin{definition}\label{NUCC}
	Let $\mathcal{F}$ be a class of functions of the form $f: \mathbb{N} \rightarrow \mathcal{A}^*$, we then call $\mathcal{F}$ an \textsl{advice class}. The non-uniform complexity class $P/\mathcal{F}$ is the class of all problems \textsl{computable in polynomial time with advice from $\mathcal{F}$}. That is, $A \in P/\mathcal{F}$, for any $A \subseteq \mathcal{A}^*$, if there exists $B \in P$ and an $f \in \mathcal{F}$ such that: $$w \in A \iff \langle w,f(|w|)\rangle \in B.$$ 
	We then say that $f$ is used to help compute $A$ in polynomial time.
\end{definition}

\begin{definition}\label{NUCC2}
	Let $f$ be an advice function, if there exists a polynomial function $p: \mathbb{N} \rightarrow \mathbb{N}$ such that $|f(n)| = O(p(n))$, then we say that $f$ is a \textsl{polynomially growing advice function}. We call the class of all polynomially growing advice functions $poly$.
	
	An advice function $g$ is a \textsl{logarithmically growing advice function} if $|g(n)| = O(\log(n))$. The class of all \textsl{logarithmically growing advice functions} is $log$.

	We call an advice function $h: \mathbb{N} \rightarrow \mathcal{A}^*$ a \textsl{prefix function} if for any $m < n$, the word $h(m)$ is a prefix word of $h(n)$. That is, for any $n \in \mathbb{N}$, by taking $h(n)$ we obtain a finite prefix word of some constant infinite sequence of symbols in $\mathcal{A}^\infty$. We denote this sequence by $h(\infty)$.

	Finally we denote the class of all \textsl{polynomially growing prefix advice functions} by $poly\star$ and the \textsl{class of all logarithmically growing prefix advice functions} by $log\star$.
\end{definition}

As shown in \cite{beggs2008oracles} that whilst $P/log \not= P/log\star$ it so happens that for computing in polynomial time restricting $poly$ to $poly\star$ does not change the non-uniform complexity class.

\begin{proposition}\label{PE}
	$P/poly = P/poly\star$.
\end{proposition}

This follows from the fact that any advice function $f: \mathbb{N} \rightarrow \mathcal{A}$ which grows at a rate of $O(n^m)$ can be substituted for a prefix advice function $g: \mathbb{N} \rightarrow (\mathcal{A} \cup \{e\})^*$ defined recursively by $g(0) = f(0)$ and $g(n+1) = g(n)ef(n+1)$ for any $n \in \mathbb{N}$. This advice function then grows at a rate of $O(n^{m+1})$.

\begin{theorem}\label{T1}
	$P_{CPCS} = P/poly$.
\end{theorem}
\begin{proof}
	$(\supseteq)$ Let $A \subseteq \{0,1\}^*$ be such that $A \in P/poly$, then by Proposition~\ref{PE}, we have $A \in P/poly\star$. Then let $g: \mathbb{N} \rightarrow \{0,1\}^*$ be a polynomially growing prefix advice function that can be used to help compute $A$, that is $|g(n)| \leqslant cn^a$ for some $c,N,a \in \mathbb{N}$ and any $n > N$.
	Define a number $\phi_g$ to have a binary expansion of $0.g(\infty)$, then consider a computation system of the form $C_{\phi_g}$, as in Example~\ref{NES}. This is a CPCS since the partition $\{[0,\frac{1}{2}),[\frac{1}{2},1)\}$ in $C_{\phi_g}$ is clearly classically measurable and the transformation $T$ it is a classically constructable since $T(x) = 2x$ if $x \in [0,\frac{1}{2})$, and $T(x) = 2x - 1$ if $x \in [\frac{1}{2},1)$.

	We can then compute $A$ in polynomial time using $C_{\phi_g}$. To do this we take a $\{0,1\}^*$-program $Q$, which on input $w \in \{0,1\}^*$ computes the first $c|w|^a$ symbols of $g(\infty)$ by computing the first $c|w|^a$ binary digits of $\phi_g$. The device then has the word $g(|w|)$ written on its tape, from which it can compute polynomially in $|w|$ whether $w \in A$. Therefore $A \in P_{CPCS}$ and thus $P_{CPCS} \supseteq P/poly$.

	$(\subseteq)$ Conversely let $C = (X,\Pi,\mathcal{T},x_0)$ be a CPCS, $C$ can be described entirely by a finite set of real numbers as follows. The multi-variable polynomial function with rational powers $F$ in Definitions~\ref{PTD} is defined entirely by the finite set of numbers $\{I,r_1,\hdots,r_I,q_{11},\hdots,q_{Im}\}$, consequently any classical transformation is defined entirely by some finite set of numbers. Similarly, balls in $\mathbb{R}^m$ are defined by finitely many real numbers that describe the coordinates of their centres and radii. Any classically measurable partition $\alpha$ in $\Pi$ is defined entirely by its finite construction from balls and classical transformations, therefore there is a finite set of real numbers $\{y_1,\hdots,y_l\}$ and a finite word $v_\alpha$ describing $\alpha$'s construction from them.
	Since $\Pi$ and $\mathcal{T}$ are a finite sets they can then also be defined in terms of a finite set of real numbers, we can similarly write the coordinates of $x_0$ as a finite set of reals. Denote the combined set of all these numbers by $D_C$ and the word describing $C$'s construction from them $v_C$.

	In general the way we obtain the digits of the elements of $D_C$ is by expanding or contracting the configuration space, and though the transformations of $\mathcal{T}$ may not commute, the digits that we can \textsl{at best} determine by a given finite sequence of transformations does not change with rearranging the order their application. This is because the digits obtainable instead depend on the rate of expansion and contraction given by the rational powers of the transformations' components. Let $p(n,k)$ be the number of ways in which we can write $n \in \mathbb{N}$ as a sum of $k$ non-negative integers, it can be shown that $p(n,1)=1$ and $p(n,k)=\sum_{l=1}^n p(l,k-1)$ for any $k \in \mathbb{N} \setminus \{0,1\}$. If the order does not matter, the number of distinct ways in which we can apply $n$ transformations from $\mathcal{T}$ is $p(n,|\mathcal{T}|)$, which grows polynomially in $n$. Therefore, the number of digits obtainable from the elements of $D_C$ grows polynomially in $n$.

	For any given polynomial function $p$, we can then encode $D_C$ as a single advice function $f_{C,p}$ which on input $n$ gives $v_C$ followed by the signed binary digits \cite{weihrauch2012computable} of each of the elements of $D_C$ that can be at best determined after $p(n)$ transformations and measurements of $C$.
	If $A$ is computable in time $O(p(n))$ using $C$, then through the polynomial advice function $f_{C,p}$ we have that $A \in P/poly$.
\end{proof}

\section{Timed Computation Systems}
\setcounter{ThmNum}{0}
One problem with our framework for computation systems thus far is that it does not adequately account for how the time a measurement takes to be carried out may vary. Indeed, as discussed by Beggs and Tucker \cite{beggs2013oracles,beggs2014analogue,beggs2012impact} the time it takes to measure a quantity may depend on the quantity itself. 

Hence we require that our computation systems be equipped with a function that describes how long it takes for a measurement to be carried out in terms of the current configuration of the computation system.

\begin{definition}
	A \textsl{timed computation system} is a 5-tuple $\mathcal{C} = (X,\Pi,\mathcal{T},x_0,\kappa)$ where $(X,\Pi,\mathcal{T},x_0)$ is a computation system and $\kappa: \Pi \times X \rightarrow \mathbb{N} \cup \{\infty\}$. We call $\kappa$ the \textsl{measurement time function} of $\mathcal{C}$.
	
\end{definition}

The value of $\kappa(\alpha,x)$ then describes how many time steps of a Turing machine it takes to measure $x$ for the partition $\alpha$. A value of $\infty$ means that this measurement takes forever.
Like with a computation system, a device acting on timed computation system is allowed access to a Turing machine. The way in which such a device acts is described as follows.

\begin{definition}
	Let $\mathcal{A}$ be an alphabet, let $\mathcal{C} = (X,\Pi,\mathcal{T},x_0,\kappa)$ be a timed computation system and let $S$ be a set of states. A \textsl{time-aware $\mathcal{A}^*$-program} $\mathcal{Q}$ on $\mathcal{C}$ and $S$ is a finite set of rules which describes how a device acts on $\mathcal{C}$ and Turing machine $TM(\mathcal{A})$.

	Each rule in $\mathcal{Q}$ takes the form $(s_i,\alpha,A,s_j,U)$ where $s_i \in S \setminus \{s_a,s_r\}$, $\alpha \in \Pi \cup \{R(\mathcal{A})\}$, $A \in \alpha \cup \emptyset$, $s_j \in S$ and $U \in \Pi \cup \mathcal{T} \cup W(\mathcal{A})$, this can be read as ``if the device is in state $s_i$ and is known to be within the subset $A$ of the partition $\alpha$ then go to state $s_j$ and commence action $U$". 

	Consider a rule of the form $(s_i,R(\mathcal{A}),A,s_j,U)$, then, as in an $\mathcal{A}^*$-program, this rule will be applied to a device in a configuration of $(x,w,s_i)$ if $w \in A$. Similarly if the action $U$ is in $\mathcal{T} \cup W(\mathcal{A})$ then a rule with this action is applied exactly if it were in an $\mathcal{A}^*$-program. Applying the rule $(s_i,\alpha,A,s_j,U)$ to a device at $(x,w,s_i)$ results in it becoming $(U(x),w,s_j)$ if $U \in \mathcal{T}$ and $(x,U(w),s_j)$ if $U \in W(\mathcal{A})$. Applying such a rule takes 1 time step.

	If the action $U$ is in $\Pi$ then it is a measurement on $X$ which may take more than a single time step to perform. When the device is at configuration $x \in X$ then the measurement takes $\kappa(U,x)$ time steps to carry out. In the mean time the device may act on its Turing tape and take single time step measurements of the tape. In this interim a rule beginning with $(s,U,A)$ can be applied if $A = \emptyset$. If the device commences a transformation or begins a new measurement on $X$ whilst the measurement of $U$ is being carried out then the process is disrupted and is ceased. If $x \in B$ for some $B \in U$, then once the measurement of $U$ is completed the output $B$ is temporarily recorded and rules beginning with $(s,U,B)$ can be applied. When the device begins a new measurement or commences a transformation on $X$, then the output $B$ is forgotten.

	In any of the above cases, if no rule that can be applied exists then the internal state becomes $s_r$. If the internal state becomes either $s_a$ or $s_r$ then the device halts.
	In order for the above process to be deterministic we once again require that, for any $(s_i,\alpha,A,s_j,T),(s_k,\beta,B,s_l,U) \in \mathcal{Q}$, if $s_i = s_k$ then $\alpha = \beta$, and if $(s_i,\alpha,A) = (s_k,\beta,B)$ then $(s_j,T) = (s_l,U)$.
\end{definition}

Just like for computation systems, we say that a problem $A \subseteq \mathcal{A}^*$ is computable using a timed computation system $\mathcal{C}$ if there exists time-aware $\mathcal{A}^*$-program $\mathcal{Q}$ on $\mathcal{C}$ that computes it. Similarly $A$ is computable using $\mathcal{C}$ in polynomial if some such $\mathcal{Q}$ has a polynomial time function.

\section{Timed Classical Physical Computation Systems, $P/poly$ and $P/log\star$}
\setcounter{ThmNum3}{0}
As in \cite{beggs2013oracles} consider a precise set of scales, we can measure whether a mass $x$ is greater or less than a given mass $y$ by placing these masses in each basket of the scale. Below each basket we place a detector which will signal if the above basket touches it. This measurement can then be described as a measurement of the partition $\alpha_y = \{[0,y],(y,\infty)\}$. Clearly, the closer the mass $x$ is to the mass $y$, the longer it will take for the scale to tip down on its heavier side onto the detector below. Indeed if $x = y$ the scales should never tip, and hence the measurement will take an infinite amount of time. In \cite{beggs2013oracles} Beggs \emph{et al.} suggest that measurement time for $\alpha_y$ on $x$ should be given by: $$\frac{d}{\sqrt{|x-y|}} < \kappa(\alpha_y,x) < \frac{c}{\sqrt{|x-y|}},$$ for some constants $c,d \in (0,\infty)$. Indeed in \cite{beggs2012impact} Beggs \emph{et al.} demonstrate that the time take for all measurements utilising a cannon and a smooth wedge should be bounded above and below by an inverse polynomial function of the distance between the position of the height of the cannon (which is what is being measured) and the cusp of the wedge.

This leads us to the following suggestion of what a reasonable requirement on the measurement time function of a CPCS should be.

\begin{definition}
	Let $\partial(A)$ be the boundary of a set $A \subseteq X$ in the Euclidean metric, the set of boundary elements of a partition $\alpha$ is then $\partial(\alpha) = \bigcup_{A \in \alpha} \partial(A)$.

	The measurement time function $\kappa$ on a classical physical computation system $C = (X,\Pi,\mathcal{T},x_0)$ takes \textsl{inverse polynomial measurement time} for any $\alpha \in \Pi$ if there exists a strictly increasing polynomial function $p: \mathbb{N} \rightarrow \mathbb{N}$ such that for any $\mathbf{x} \in X$: $$\kappa(\alpha,\mathbf{x}) = p\left(\left\lceil\frac{1}{\inf_{\mathbf{y} \in \partial(\alpha)} |\mathbf{x} - \mathbf{y}|}\right\rceil\right).$$
\end{definition}

\begin{definition}
	A timed computation system $\mathcal{C} = (X,\Pi,\mathcal{T},x_0,\kappa)$ is a \textsl{timed classical physical computation system} (TCPCS) if $(X,\Pi,\mathcal{T},x_0)$ is a classical physical computation system and $\kappa$ takes at least inverse polynomial measurement time.
\end{definition}

Denote the class of problems computable using a TCPCS in polynomial time by $P_{TCPCS}$. It so happens though that TCPCS's have the same computation power in polynomial time as CPCS's.

\begin{theorem}
	$P_{TCPCS} = P/poly$
\end{theorem}
%
%
\begin{proof}
	$(\supseteq)$ For any problem $A \subseteq \{0,1\}^*$ such that $A \in P/poly$, let $g: \mathbb{N} \rightarrow \{0,1\}^*$ be a polynomially growing prefix advice function that can be used to help compute $A$, that is $|g(n)| \leqslant cn^a$ for some $c,N,a \in \mathbb{N}$ any $n > N$.
	Define $\mathcal{C}_g = ([0,1),\{\alpha\},\{T\},\psi_g,\kappa)$ to be a TCPCS where $\alpha = \{[0,\frac{1}{3}),[\frac{1}{3},\frac{2}{3}),[\frac{2}{3},1)\}$, and $T(x) = 3x - \lfloor 3x \rfloor$. The number $\psi_g$ has a ternary expansion that consists of the infinite word $g(\infty)$ with an extra $2$ inserted after every digit. Also $\kappa(\alpha,x) = \max_{k \in \{0,1,2,3\}}|x-\frac{k}{3}|^{-1}$, for each $x \in X$.

	We have constructed $\psi_g$ such that every other digit is a $2$ in order to ensure that $T^l(\psi_g)$, for any $l \in \mathbb{N}$, is always bounded away from the boundary of each element of $\alpha$. We will now consider separately the case where $l$ is even and the case where $l$ is odd. Suppose that $l$ is even, then $T^l(\psi_g) \in [\frac{k}{3},\frac{k+1}{3})$ where $k = 0$ or 1, as any even digit of $\psi_g$ is always a 0 or a 1. Further more, $T^l(\psi_g) \not\in [\frac{k}{3},\frac{k+1}{3}-\frac{1}{9})$ as any odd digit of $\psi_g$ is a 2 and so $T^{l+1}(\psi_g) \in [\frac{2}{3},1)$. Similarly $T^l(\psi_g) \not\in [\frac{k+1}{3}-\frac{1}{27},\frac{k+1}{3})$ as $T^{l+2}(\psi_g) \in [0,\frac{2}{3})$. Therefore if $l$ is even $|T^l(\psi_g) - \frac{k}{3}| > \frac{1}{27}$ for any $k \in \{0,1,2,3\}$. Now if $l$ is odd then $T^{l}(\psi_g) \in [\frac{2}{3},1)$, further more $T^l(\psi_g) \not\in [\frac{8}{9},1)$ since $T^{l+1}(\psi_g) \in [0,\frac{2}{3})$, and $T^l(\psi_g) \not\in [\frac{2}{3},\frac{20}{27})$ since $T^{l+2}(\psi_g) \in [\frac{2}{3},1)$. Therefore $|T^l(\psi_g) - \frac{k}{3}| > \frac{2}{27}$ for any $k \in \{0,1,2,3\}$. Hence for any $T^l(\psi_g)$ the time taken to perform an $\alpha$ measurement on it is at most $\frac{1}{1/27} = 27$ time steps.

	We can then compute $A$ in polynomial time using $\mathcal{C}_g$ via a time-aware $\mathcal{A}^*$-program $\mathcal{Q}$ similar to the program $Q$ in the proof of Theorem~\ref{T1}. The program $\mathcal{Q}$ computes the first $cn^a$ symbols of $g(\infty)$ by determining the first $2cn^a$ ternary digits of $\psi_g$ in $\mathcal{C}_g$, it then computes $A$ in polynomial time using this advice. As each transformation takes 1 step and each measurement takes at most 27 steps, and so to determine the first $2cn^a$ ternary digits of $\psi_g$ takes at most $56cn^a$ steps, which is polynomial in $n$.

	$(\subseteq)$ Conversely, let $\mathcal{C} = (X,\Pi,\mathcal{T},x_0,\kappa)$ be a TCPCS. Then a CPCS defined to by $C = (X,\Pi,\mathcal{T},x_0)$ can compute any problem as quickly as $\mathcal{C}$, unless $\mathcal{C}$ is able to obtain some information by performing a measurement $\alpha$ on some $x \in X$ and counting how many steps the measurement took. Hence we need to define a CPCS in which every partition of the form: $$\beta_{\alpha,n} = \{\{x \in X \text{ } | \text{ } \kappa(\alpha,x) < n\},\{x \in X \text{ } | \text{ } \kappa(\alpha,x) \geqslant n\}\},$$ can be determined in $q(n)$ times steps for some polynomial function $q$. 

	Let $C' = (X \times \mathbb{R},\Pi',\mathcal{T}' \cup \mathcal{U} \cup \mathcal{V} \cup \{f^+,f^-\},(x_0,0))$ be a CPCS. For every open ball $B_\epsilon(\mathbf{x})$ used in the construction of some classically measurable set, in some partition of $\Pi$, the set $\Pi'$ contains the partition $\{B_\epsilon(\mathbf{x}) \times \mathbb{R},(X \times \mathbb{R}) \setminus (B_\epsilon(\mathbf{x}) \times \mathbb{R})\}$. Similarly for every closed ball, used in some classically measurable set construction, $\Pi'$ contains $\{\overline{B}_\epsilon(\mathbf{x}) \times \mathbb{R},(X \times \mathbb{R}) \setminus (\overline{B}_\epsilon(\mathbf{x}) \times \mathbb{R})\}$. We can extend any $T \in \mathcal{T}$ to a classical transformation $T'$ on $X \times \mathbb{R}$ such that $T'(\mathbf{x},y) = (T(\mathbf{x}),y)$ for any $(\mathbf{x},y) \in X \times \mathbb{R}$. So let $T' \in \mathcal{T}'$, if $T \in \mathcal{T}$. Let $U' \in \mathcal{U}$ if $U$ is a classical transformation used in the construction of an element of a partition of $\Pi$. Hence through $\Pi'$ and $\mathcal{U}$ a device acting on $C'$ can perform any measurement of $\Pi$ in a fixed finite number of time steps. For any $(\mathbf{x},y) \in X \times \mathbb{R}$ we let $f^+(\mathbf{x},y) = (\mathbf{x},y+1)$ and $f^-(\mathbf{x},y) = (\mathbf{x},y-1)$. For every ball of centre $\mathbf{z}$ in a partition of $\Pi'$ we have a transformation $h_\mathbf{z} \in \mathcal{V}$ which is such that for any $(\mathbf{x},y) \in X \times \mathbb{R}$: $$h_\mathbf{z}(\mathbf{x},y) = \left(\mathbf{x}-\frac{\mathbf{x}-\mathbf{z}}{y|\mathbf{x}-\mathbf{z}|},y\right).$$
	Applying $h_\mathbf{z}$ to the ball $B_\epsilon(\mathbf{z})$ results in $B_{\epsilon+\frac{1}{y}}(\mathbf{z})$. Finally, as $\kappa$ takes inverse polynomial measurement time for any $\alpha \in \Pi$, the set of $x \in X$ such that $\kappa(\alpha,x) < p(n)$, corresponds to the set of elements at least a distance of $\frac{1}{n}$ from the boundaries of $\alpha$. We can therefore determine $\beta_{\alpha,n}$ by constructing it from balls of the form $B_{\epsilon+\frac{1}{y}}(\mathbf{z})$ and the transformations of $\mathcal{U}$.

	Therefore, any problem computable in polynomial time by $\mathcal{C}$ is computable in polynomial time by $C'$ and so $P_{TCPCS} \subseteq P_{CPCS}$. Consequently $P_{TCPCS} \subseteq P/poly$ by Theorem~\ref{T1}. 
\end{proof}

The above result is somewhat contradictory to the suggested computational power of $P/log\star$ for classical physical systems with infinite precision given by Beggs \emph{et al.} in \cite{beggs2014analogue}, since from \cite{beggs2008oracles} we know that $P/log\star \subsetneq P/poly$. There is a remedy though, as the computational power in $\mathcal{C}_g$ above came from its initial configuration, not its measurements. If the initial configuration of $\mathcal{C}_g$ was an algebraic number\footnote{We denote the set of real algebraic numbers by $\mathbb{A}$.} then its digits would be computable in polynomial time on a Turing machine and thus a device acting on $\mathcal{C}_g$ would not be computationally more powerful. However, we would still be able to obtain the same computational power as $\mathcal{C}_g$ above if we were able to apply a transformation that mapped an algebraic element of $[0,1)$ to an arbitrary real number $\phi_g$. We thus have the following restriction.

\begin{definition}
	A timed classical physical computation system $\mathcal{C} = (X,\Pi,\mathcal{T},x_0,\kappa)$ with $X \subseteq \mathbb{R}^m$, is \textsl{algebraically acting} if $x_0 \in \mathbb{A}^m$, and every $T \in \mathcal{T}$ is classical transformation defined on the whole of $X$ such that the coefficients of every multi-variable polynomial function with rational powers used in its construction are all in $\mathbb{A}$.
\end{definition}

Denote the class of problems computable using an algebraically acting TCPCS in polynomial time by $P_{ATCPCS}$. This restriction is arguably natural as square matrices with algebraic coefficients are examples of algebraically acting classical transformations. So are rational rotations, as $\cos(\theta)$ and $\sin(\theta)$ are algebraic numbers if $\theta$ is a rational multiple of $\pi$.

It should be noted that if we were to restrict our systems to having only transformations with algebraic coefficients whilst allowing for the initial configuration $x_0$ to be an arbitrary element of $\mathbb{R}^m$, the result is a computation system with the same computational power in polynomial time as general TCPCS's. This is because such a device is able to use the coordinates of $x_0$ as a resource of arbitrary real coefficients, effectively allowing the device to carry out arbitrary classical transformations. However, if we restrict the initial configuration as well, we have the following result.

\begin{theorem}\label{TAIT}
	$P_{ATCPCS} = P/log\star$.
\end{theorem}
%
%
\begin{proof}
	$(\supseteq)$ For any problem $A \subseteq \{0,1\}^*$ such that $A \in P/log\star$, let $g: \mathbb{N} \rightarrow \{0,1\}^*$ be a logarithmically growing prefix advice function that can be used to help compute $A$, say for some $c,N \in \mathbb{N}$ any $n > N$, $|g(n)| \leqslant c\log(n)$. Define $\mathcal{D}_g = ([0,1],\{\alpha_g\},\{T_0,T_1,R\},\frac{1}{2},\kappa)$ to be an algebraically acting TCPCS where $\alpha_g = \{[0,\phi_g],(\phi_g,1]\}$ with the number $\phi_g$ having a binary expansion of $0.g(\infty)$. The transformations are $T_0(x) = \frac{x}{2}$, $T_1(x) = \frac{x+1}{2}$, and  $R(x) = \frac{1}{2}$ for any $x \in [0,1)$. The measurement time function is such that $\kappa(\alpha_g,x) = |x-\phi_g|^{-1}$ for each $x \in [0,1)$.

	Let $y \in [0,1)$ have a binary expansion of $0.b_1b_2\hdots$. Applying $T_0$ to $y$ gives a number with a binary expansion of $0.0b_1b_2\hdots$, and applying $T_1$ to $y$ gives a number with a binary expansion of $0.1b_1b_2\hdots$. Hence by repeated applications of $T_0$ and $T_1$ to $0$ we can approximate any number in $[0,1)$ with a finite binary expansion.

	We can therefore determine $\phi_g$ to arbitrarily many places in $\mathcal{D}_g$. To see this, suppose that we know that the first $l$ binary digits of $\phi_g$ are $a_1\hdots a_l$, we can then generate a number $z_l$ with a binary expansion of of $0.a_1\hdots a_l1$ and apply the measurement $\alpha_g$ to it. If we learn that $z_l \in (\phi_g,1]$ then we know that $\phi_g < z_l$ and so the first $l+1$ binary digits of $\phi_g$ are $a_1\hdots a_l 0$. If we learn that $z_l \in [0,\phi_g]$ then we know that $\phi_g \geqslant z_l$ and hence the first $l+1$ binary digits of $\phi_g$ are $a_1\hdots a_l 1$. Additionally, if $\kappa(\alpha_g,z_l) > 2^{l+1}$ then we know that $|z_l-\phi_g| \leqslant 2^{-(l+1)}$ and thus the first $l+1$ binary digits of $\phi_g$ must be $a_1\hdots a_l1$. Once we know this, we can reset the system to $\frac{1}{2}$ using $R$ and generate a new number with a more accurate binary expansion. Using this process we can determine $\phi_g$ to $L$ binary places in $O(\sum_{j=1}^L 2^{j}) = O(2^L)$ time steps.

	On input $w \in \{0,1\}^*$ we can then determine $c\log(|w|)$ places of $g(\infty)$ in polynomial time using the above procedure,as $O(2^{c\log|w|}) = O(|w|^c)$. 
	Hence we can compute $A$ in polynomial time using $\mathcal{D}_g$ via a $\{0,1\}^*$-program that determines the first $c\log|w|$ symbols of $g(\infty)$ before computing $A$ using this advice.

	$(\subseteq)$ Conversely, let $\mathcal{C} = (X,\Pi,\mathcal{T},x_0,\kappa)$ be an algebraically acting TCPCS, as in the proof of Theorem~\ref{T1} we can construct a logarithmically growing prefix advice function $f_{\mathcal{C}}$ that encodes the finite description of $\mathcal{T}$, the polynomial functions that define the measurement times of each $\alpha \in \Pi$ for $\kappa$ and $x_0$ in $f_\mathcal{C}(0)$. The prefix advice function $f_\mathcal{C}$ gives $f_\mathcal{C}(0)$ followed by logarithmically many digits of each of the real numbers that are used to define the elements of the partitions in $\Pi$. 

	Then any problem $A$ computable in polynomial time using $\mathcal{C}$ is computable in polynomial time with logarithmically growing advice $f_{\mathcal{C}}$. This is because the only information obtainable from $\mathcal{C}$ that is not given immediately by $f_\mathcal{C}$ is within $\Pi$, and so in order to obtain $l$ signed binary digits of a real number used to define an element $A \in \alpha$ of some partition $\alpha \in \Pi$ the device needs to be in a configuration that is within a distance of $2^{-l}$ of the boundary of $A$. Measuring this configuration takes at least $2^l$ time steps. If a partition is useful for the computation of $A$ then we also cannot use any transformation in $\mathcal{T}$ to reliably reduce this measurement time. As for example, if $X = \mathbb{R}$ and $\alpha = \{(-\infty,\phi],(\phi,\infty)\}$ then a transformation $T$ that maps numbers around $\phi$ away from it whilst keeping these numbers within the same elements of $\alpha$ must have a fixed point at $\phi$ since by the construction $T$ is continuous. But then by the construction of $T$, the number $\phi$ must be algebraic and hence computable in polynomial time by a finite program on a Turing machine. Similarly for a general $X \subseteq \mathbb{R}^m$ if we can transform away from the boundary of some partition of $X$ using a algebraic classical transformation then that part of the boundary must be algebraically defined and hence computable in polynomial time. 
\end{proof}

We thus obtain the Beggs \emph{et al.}'s suggested computational power for infinite precision analogue-digital devices by preventing our TCPCS's from applying transcendental transformations. Indeed, a key point to take from the proof of Theorem~\ref{TAIT} is that requiring that measurements take inverse polynomial time is only a restriction to the computational power of an infinite precision classical physical system if its transformations and initial configuration are also restricted to being computable in polynomial time on $\mathbb{R}^m$ with the Euclidean topology \cite{weihrauch2012computable}. We of course chose our transformations to be constructed from multi-variable polynomial functions with rational powers as per the reasoning at the start of Section 3, we do not know of any physical justifications for allowing a more general extension. However, the TCPCS we used in the $(\supseteq)$ part of our proof of Theorem~\ref{TAIT} uses only rational coefficients, so we could further restrict our classical transformations to being only rationally acting. An extension of this result to a class of differentiable manifolds should also be possible.

\bibliographystyle{eptcs}
\bibliography{APC}
\end{document}